\documentclass[11pt]{article}%
\usepackage{amsfonts}
\usepackage{amsmath}
\usepackage{amssymb}
\usepackage{graphicx}%
\providecommand{\U}[1]{\protect\rule{.1in}{.1in}}
\newtheorem{theorem}{Theorem}
\newtheorem{acknowledgement}[theorem]{Acknowledgement}

\newtheorem{proposition}[theorem]{Proposition}

\newenvironment{proof}[1][Proof]{\noindent\textbf{#1.} }{\ \rule{0.5em}{0.5em}}
\begin{document}

\title{The Phase Space Formulation of Time-Symmetric Quantum Mechanics, I: the Wigner Formalism}
\author{Charlyne de Gosson \ and \ Maurice de Gosson \ \\University of Vienna, NuHAG Faculty of Mathematics\\Oskar-Morgenstern-Platz 1, 1090 Vienna AUSTRIA}
\maketitle

\begin{abstract}
Time-symmetric quantum mechanics can be described in the usual
Weyl--Wigner--Moyal formalism (WWM) by using the properties of the Wigner
distribution, and its generalization, the cross-Wigner distribution. The use
of the latter makes clear a strongly oscillating interference between the pre-
and post-selected states. This approach allows us to give explicit formulas
for the state reconstruction problem, thus generalizing known results to the
case of arbitrary observables. In a forthcoming paper we will extend these
results to other quantization schemes.

\end{abstract}

\section{Introduction and Description of the Problem}

We will work with systems having $n$ degrees of freedom. Position (resp.
momentum) variables are denoted $x=(x_{1},...,x_{n})$ (resp. $p=(p_{1}%
,...,p_{n})$). The corresponding phase space variable is $(x,p)$. The scalar
product $p_{1}x_{1}+\cdot\cdot\cdot+p_{n}x_{n}$ is denoted by $px$. When
integrating we will use, where appropriate, the volume elements $d^{n}%
x=dx_{1}\cdot\cdot\cdot dx_{n}$, $d^{n}p=dp_{1}\cdot\cdot\cdot dp_{n}$. The
unitary $\hbar$-Fourier transform of a square-integrable function $\Psi$ of
$x$ is
\begin{equation}
\widehat{\Psi}(p)=\left(  \tfrac{1}{2\pi\hbar}\right)  ^{n/2}\int e^{-\frac
{i}{\hbar}px}\Psi(x)d^{n}x.\label{FT}%
\end{equation}
We denote by $\widehat{x}=(\widehat{x}_{1},...,\widehat{x}_{n})$ and
$\widehat{p}=(\widehat{p}_{1},...,\widehat{p}_{n})$ the (vector) operators
defined by $\widehat{x}_{j}\Psi=x_{j}\Psi$, $\widehat{p}_{j}\Psi
=-i\hbar\partial_{x_{j}}\Psi$.

\subsection{The notion of weak value}

In time-symmetric quantum mechanics (TSQM) the state of a system is
represented by a two-state vector $\langle\Phi|~|\Psi\rangle$ where the state
$\langle\Phi|$ evolves backwards from the future and the state $|\Psi\rangle$
evolves forwards from the past. To make things clear, assume that at a time
$t_{\mathrm{i}}$ an observable $\widehat{A}$ is measured and a non-degenerate
eigenvalue was found: $|\Psi(t_{\mathrm{i}})\rangle=|\widehat{A}=\alpha
\rangle$; similarly at a later time $t_{\mathrm{f}}$ a measurement of another
observable $\widehat{B}$ yields $|\Phi(t_{\mathrm{f}})\rangle=|\widehat
{B}=\beta\rangle$. Such a two-time state $\langle\Phi|~|\Psi\rangle$ can be
created as follows \cite{ABL,silva}: Alice prepares a state $|\Psi
(t_{\mathrm{i}})\rangle$ at initial time $t_{\mathrm{i}}$. She then sends the
system to an observer, Bob, who may perform any measurement he wishes to. The
system is returned to Alice, who then performs a$,$ strong measurement with
the state $|\Phi(t_{\mathrm{f}})\rangle$ as one of the outcomes. Only if this
outcome is obtained, does Bob keep the results of his measurement.

Let now $t$ be some intermediate time: $t_{\mathrm{i}}<t<t_{\mathrm{f}}$.
Following the time-symmetric approach to quantum mechanics at this
intermediate time the system is described by the \emph{two} wavefunctions%
\begin{equation}
\Psi=U_{\mathrm{i}}(t,t_{\mathrm{i}})\Psi(t_{\mathrm{i}})\text{ \ , \ }%
\Phi=U_{\mathrm{f}}(t,t_{\mathrm{f}})\Phi(t_{\mathrm{f}})\label{1}%
\end{equation}
where $U_{\mathrm{i}}(t,t^{\prime})=e^{-i\widehat{H_{\mathrm{i}}}(t-t^{\prime
})/\hbar\text{ }}$ and $U_{\mathrm{f}}(t,t^{\prime})=e^{-i\widehat
{H_{\mathrm{f}}}(t-t^{\prime})/\hbar\text{ }}$ are the unitary operators
governing the evolution of the state before and after time $t$. Consider now
the superposition of the two states $|\Psi\rangle$ and $|\Phi\rangle$ (which
we suppose normalized); the expectation value
\[
\langle\widehat{A}\rangle_{\Psi+\Phi}=\langle\Psi+\Phi|\widehat{A}|\Psi
+\Phi\rangle
\]
of the observable $\widehat{A}$ in this superposition is obtained using the
equality%
\begin{equation}
||\Phi+\Psi||\langle\widehat{A}\rangle_{\Psi+\Phi}=\langle\widehat{A}%
\rangle_{\Phi}+\langle\widehat{A}\rangle_{\Psi}+2\operatorname{Re}\langle
\Phi|\widehat{A}|\Psi\rangle;\label{3}%
\end{equation}
setting $N=||\Phi+\Psi||$ we get, assuming $\langle\Phi|\Psi\rangle\neq0$,
\begin{equation}
\langle\widehat{A}\rangle_{\Psi+\Phi}=\frac{1}{N}\left(  \langle\widehat
{A}\rangle_{\Phi}+\langle\widehat{A}\rangle_{\Psi}+2\operatorname{Re}%
(\langle\Phi|\Psi\rangle\langle\widehat{A}\rangle_{\Phi,\Psi}\right)
\label{4}%
\end{equation}
where
\begin{equation}
\langle\widehat{A}\rangle_{\Phi,\Psi}=\frac{\langle\Phi|\widehat{A}%
|\Psi\rangle}{\langle\Phi|\Psi\rangle}\label{5}%
\end{equation}
is, by definition, the \emph{weak value} of $\widehat{A}$. Weak values provide
an unexpected insight into a number of of fundamental quantum effects.

We will assume from now on that $|\Psi\rangle$ and $|\Phi\rangle$ are two
normalized non-orthogonal states: $\langle\Psi|\Psi\rangle=\langle\Phi
|\Phi\rangle=1$, $\langle\Phi|\Psi\rangle\neq0$.

\subsection{What we will do}

In the discussion above we have been working directly in terms of the
wavefunctions $\Psi$ and $\Phi$; now, a different kind of state description
which is very fruitful, particularly in quantum optics, is provided by the
Wigner distribution
\cite{Wigner,Folland,Birkbis,Springer,hai,hillery,Littlejohn}
\begin{equation}
W_{\Psi}(x,p)=\left(  \tfrac{1}{2\pi\hbar}\right)  ^{n}\int e^{-\frac{i}%
{\hbar}py}\Psi(x+\tfrac{1}{2}y)\Psi^{\ast}(x-\tfrac{1}{2}y)d^{n}y;\label{6}%
\end{equation}
the latter is directly related to the mean value $\langle\widehat{A}%
\rangle_{\Psi}=\langle\Psi|\widehat{A}|\Psi\rangle$ by Moyal's formula
\cite{hillery,Moyal,Folland,Birkbis,Springer}%
\begin{equation}
\langle\widehat{A}\rangle_{\Psi}=\iint a(x,p)W_{\Psi}(x,p)d^{n}pd^{n}%
x\label{7}%
\end{equation}
where $a(x,p)$ is the classical observable whose Weyl quantization is given by
the Weyl--Moyal formula%
\begin{equation}
\widehat{A}=\left(  \tfrac{1}{2\pi\hbar}\right)  ^{n}\iint\widehat
{a}(x,p)e^{\frac{i}{\hbar}(x\widehat{x}+p\widehat{p})}d^{n}pd^{n}x\label{7bis}%
\end{equation}
(we use the terminology classical observable\textquotedblright\ in a very
broad sense; $a$ can be any complex integrable function, or even a tempered
distribution, \textit{i.e.} an element of $\mathcal{S}^{\prime}(\mathbb{R}%
^{2n})$, dual of the Schwartz space $\mathcal{S}(\mathbb{R}^{2n})$ of rapidly
decreasing functions). A direct calculation shows that we have
\begin{equation}
W_{\Psi+\Phi}=W_{\Phi}+W_{\Psi}+2\operatorname{Re}W_{\Psi,\Phi}\label{8}%
\end{equation}
where the cross-term $W_{\Psi,\Phi}$ is given by
\begin{equation}
W_{\Psi,\Phi}(x,p)=\left(  \tfrac{1}{2\pi\hbar}\right)  ^{n}\int e^{-\frac
{i}{\hbar}py}\Psi(x+\tfrac{1}{2}y)\Phi^{\ast}(x-\tfrac{1}{2}y)d^{n}y.\label{9}%
\end{equation}
The appearance of the term $W_{\Psi,\Phi}$ shows the emergence at time $t$ of
a strong interference between the preselected and the post-selected states
$|\Psi\rangle$ and $|\Phi\rangle$. It is called the cross-Wigner distribution
of $\Psi,\Phi$, see \cite{Folland,Birkbis,hla} and the references therein. We
are going to exploit the properties of $W_{\Psi,\Phi}$ to give an alternative
working definition of the weak value $\langle\widehat{A}\rangle_{\Phi,\Psi}$,
namely
\begin{equation}
\langle\widehat{A}\rangle_{\Phi,\Psi}=\frac{1}{\langle\Phi|\Psi\rangle}\iint
a(x,p)W_{\Psi,\Phi}(x,p)d^{n}pd^{n}x
\end{equation}
(formula (\ref{13})); here $a(x.p)$ is the classical observable whose Weyl
quantization is the operator $\widehat{A}$. This allows the function
\begin{equation}
\rho_{\Phi,\Psi}(x,p)=\frac{W_{\Psi,\Phi}(x,p)}{\langle\Phi|\Psi\rangle}%
\end{equation}
to be interpreted as a complex probability distribution. We thereafter notice
that the cross-Wigner distribution can itself be seen, for fixed $(x,p)$, as a
weak value, namely that of Grossmann and Royer's parity operator $\widehat
{T}_{\text{GR}}(x,p)$:%
\begin{equation}
W_{\Psi,\Phi}(x,p)=(\pi\hbar)^{n}\langle\widehat{T}_{\text{GR}}(x,p)\rangle
_{\Psi,\Phi}\langle\Phi|\Psi\rangle
\end{equation}
(formula (\ref{woofy1})). Using this approach we prove (formula (\ref{13.24}))
the following reconstruction formula: if $W_{\Psi,\Phi}$ is known, we can
reconstruct (up to an unessential phase factor) the wave function $\Psi$ (and
hence the state $|\Psi\rangle$) using the formula
\begin{equation}
\Psi(x)=\frac{2^{n}}{\langle\Phi|\Lambda\rangle}\iint W_{\Psi,\Phi
}(y,p)\widehat{T}_{\text{GR}}(y,p)\Lambda(x)d^{n}pd^{n}y
\end{equation}
where $\Lambda$ is an arbitrary square-integrable function such that
$\langle\Phi|\Lambda\rangle\neq0$. 

\section{Weak Values in the Wigner Picture}

\subsection{The cross-Wigner transform}

The cross-Wigner distribution is defined for all square-integrable functions
$\Psi,\Phi$; it satisfies the generalized marginal conditions%
\begin{align}
\int W_{\Psi,\Phi}(x,p)d^{n}p &  =\Psi(x)\Phi^{\ast}(x)\label{16}\\
\int W_{\Psi,\Phi}(x,p)d^{n}x &  =\widehat{\Psi}(p)\widehat{\Phi}^{\ast
}(p)\label{17}%
\end{align}
provided that $\Psi$ and $\Phi$ are in $L^{1}(\mathbb{R}^{n})\cap
L^{2}(\mathbb{R}^{n})$; these formulas reduce to the usual marginal conditions
for the Wigner distribution when $\Psi=\Phi$. While $W_{\Psi}$ is always real
(though not non-negative, unless $\Psi$ is a Gaussian), $W_{\Psi,\Phi}$ is a
complex function, and we have $W_{\Psi,\Phi}^{\ast}=W_{\Phi,\Psi}$. The
cross-Wigner distribution is widely used in signal theory and time-frequency
analysis \cite{Folland,hla}; its Fourier transform is the cross-ambiguity
function familiar from radar theory \cite{Folland,Woodward,szu}. Zurek
\cite{Zurek} has studied $W_{\Psi,\Phi}$ when $\Psi+\Phi$ is a Gaussian
cat-like state, and shown that it is accountable for sub-Planck structures in
phase space due to interference.

We now make the following elementary, but important remark: multiplying both
sides of the equality (\ref{8}) by the classical observable $a(x,p)$ and
integrating with respect to the $x,p$ variables, we get, using Moyal's formula
(\ref{7}),
\begin{equation}
||\Phi+\Psi||\langle\widehat{A}\rangle_{\Psi+\Phi}=\langle\widehat{A}%
\rangle_{\Phi}+\langle\widehat{A}\rangle_{\Psi}+2\iint a(x,p)\operatorname{Re}%
W_{\Psi,\Phi}(x,p)d^{n}pd^{n}x.\label{10}%
\end{equation}
Comparing with formula (\ref{4}) we see that%
\begin{equation}
\operatorname{Re}\langle\Phi|\widehat{A}|\Psi\rangle=\iint
a(x,p)\operatorname{Re}W_{\Psi,\Phi}(x,p)d^{n}pd^{n}x.\label{11}%
\end{equation}
It turns out that in the mathematical theory of the Wigner distribution
\cite{Folland,Birkbis} one shows that the equality above actually holds not
only for the real parts, but also for the purely imaginary parts, hence we
always have%
\begin{equation}
\langle\Phi|\widehat{A}|\Psi\rangle=\iint a(x,p)W_{\Psi,\Phi}(x,p)d^{n}%
pd^{n}x.\label{12}%
\end{equation}
An immediate consequence of this equality is that we can express the weak
value $\langle\widehat{A}\rangle_{\Phi,\Psi}$ in terms of the cross-Wigner
distribution and the classical observable $a(x,p)$ corresponding to
$\widehat{A}$ in the \emph{Weyl quantization scheme}:
\begin{equation}
\langle\widehat{A}\rangle_{\Phi,\Psi}=\frac{1}{\langle\Phi|\Psi\rangle}\iint
a(x,p)W_{\Psi,\Phi}(x,p)d^{n}pd^{n}x.\label{13}%
\end{equation}
We emphasize that one has to be excessively careful when using formulas of the
type (\ref{13}) (as we will do several times in this work): the function $a$
crucially depends on the quantization procedure which is used (here Weyl
quantization); we will come back to this essential point later, but here is a
simple example which shows that things can get wrong if this rule is not
observed: let $\widehat{H}=\frac{1}{2}(\widehat{x}^{2}+\widehat{p}^{2})$ be
the quantization of the normalized harmonic oscillator $H(x,p)=\frac{1}%
{2}(x^{2}+p^{2})$ (we assume $n=1$). While it is true that
\begin{equation}
\langle\widehat{H}\rangle_{\Phi,\Psi}=\frac{1}{\langle\Phi|\Psi\rangle}\iint
H(x,p)W_{\Psi,\Phi}(x,p)dpdx\label{13a}%
\end{equation}
it is in contrast \emph{not true} that
\begin{equation}
\langle\widehat{H}^{2}\rangle_{\Phi,\Psi}=\frac{1}{\langle\Phi|\Psi\rangle
}\iint H(x,p)^{2}W_{\Psi,\Phi}(x,p)dpdx.\label{13b}%
\end{equation}
Suppose for instance that $\Psi=\Phi$ is the ground state of the harmonic
oscillator: $\widehat{H}\Psi=\frac{1}{2}\hbar\Psi$. We have $\langle
\widehat{H}^{2}\rangle-\langle\widehat{H}\rangle^{2}=0$; however use of
formula (\ref{13b}) yields the wrong result $\langle\widehat{H}^{2}%
\rangle-\langle\widehat{H}\rangle^{2}=\frac{1}{4}\hbar^{2}$. The
error\footnote{It is remarkable that similar errors still appear in many
texts, even the best (see e.g. \cite{27}).} comes from the inobservance of the
prescription above: $\widehat{H}^{2}$ is not the Weyl quantization of
$H(x,p)^{2}$, but that of $H(x,p)^{2}-\frac{1}{4}\hbar^{2}$ as is easily seen
using the McCoy \cite{coy} rule%
\begin{equation}
\widehat{x^{r}p^{s}}=\frac{1}{2^{s}}\sum_{k=0}^{s}\binom{s}{k}\widehat
{p}^{s-k}\widehat{x}^{r}\widehat{p}^{k}\label{7ter}%
\end{equation}
and Born's canonical commutation relation $[\widehat{x},\widehat{p}]=i\hbar$
(see Shewell \cite{Shewell} for a discussion of related, examples).

\subsection{A complex phase space distribution}

Let us now set%
\begin{equation}
\rho_{\Phi,\Psi}(x,p)=\frac{W_{\Psi,\Phi}(x,p)}{\langle\Phi|\Psi\rangle};
\label{rho1}%
\end{equation}
using the marginal conditions (\ref{16})--(\ref{17}) we get
\begin{align}
\int\rho_{\Phi,\Psi}(x,p)d^{n}p  &  =\frac{\Phi^{\ast}(x)\Psi(x)}{\langle
\Phi|\Psi\rangle}\label{rhomarg1}\\
\int\rho_{\Phi,\Psi}(x,p)d^{n}x  &  =\frac{\widehat{\Phi}^{\ast}%
(p)\widehat{\Psi}(p)}{\langle\Phi|\Psi\rangle} \label{rhomarg2}%
\end{align}
hence the function $\rho_{\Phi,\Psi}$ is a complex probability distribution:
\begin{equation}
\int\rho_{\Phi,\Psi}(x,p)d^{n}pd^{n}x=1. \label{rhonorm}%
\end{equation}
The weak value is given in terms of $\rho_{\Phi,\Psi}$ by%
\begin{equation}
\langle\widehat{A}\rangle_{\Phi,\Psi}=\int a(x,p)\rho_{\Phi,\Psi}%
(x,p)d^{n}pd^{n}x \label{fundamental}%
\end{equation}
which reduces to the usual formula (\ref{7}) in the case of an ideal
measurement (\textit{i.e}. $\Phi=\Psi$). The practical meaning of these
relations is the following (\cite{AV3}, Chapter 13): the readings of the
pointer of the measuring device will cluster around the value
\begin{equation}
\operatorname{Re}\langle\widehat{A}\rangle_{\Phi,\Psi}=\int\operatorname{Re}%
(a(x,p)\rho_{\Phi,\Psi}(x,p))d^{n}pd^{n}x \label{rerho}%
\end{equation}
while the quantity%
\begin{equation}
\operatorname{Im}\langle\widehat{A}\rangle_{\Phi,\Psi}=\int\operatorname{Im}%
(a(x,p)\rho_{\Phi,\Psi}(x,p))d^{n}pd^{n}x \label{imro}%
\end{equation}
measures the shift in the variable conjugate to the pointer variable. In an
interesting paper \cite{feyer} Feyereisen discusses some aspects of the
complex distribution $\rho_{\Phi,\Psi}$.

\subsection{The cross-Wigner transform as a weak value}

Let $\widehat{T}(x_{0},p_{0})=e^{-\frac{i}{\hbar}(p_{0}\widehat{x}%
-x_{0}\widehat{p})}$ be the Heisenberg operator; it is a unitary operator
whose action on a wavefunction $\Psi$ is given by
\begin{equation}
\widehat{T}(x_{0},p_{0})\Psi(x)=e^{\tfrac{i}{\hslash}(p_{0}x-\tfrac{1}{2}%
p_{0}x_{0})}\Psi(x-x_{0}). \label{HW}%
\end{equation}
It has the following simple dynamical interpretation \cite{Birkbis,Littlejohn}%
: $\widehat{T}(z_{0})$ is the time-one propagator for the Schr\"{o}dinger
equation corresponding to the translation Hamiltonian $H_{0}=x_{0}p-p_{0}x$.
An associated operator is the Grossmann--Royer reflection operator (or
displacement parity operator) \cite{Birkbis,Grossmann,Royer}:%
\begin{equation}
\widehat{T}_{\text{GR}}(x_{0},p_{0})=\widehat{T}(x_{0},p_{0})R^{\vee}%
\widehat{T}(x_{0},p_{0})^{\dagger} \label{GR}%
\end{equation}
where $R^{\vee}$ changes the parity of the function to which it is applied:
$R^{\vee}\Psi(x)=\Psi(-x)$; the explicit action of $\widehat{T}_{\text{GR}%
}(z_{0})$ on wavefunctions is easily obtained using formula (\ref{HW}) and one
finds%
\begin{equation}
\widehat{T}_{\text{GR}}(x_{0},p_{0})\Psi(x)=e^{\frac{2i}{\hbar}p_{0}(x-x_{0}%
)}\Psi(2x_{0}-x). \label{GRbis}%
\end{equation}
Now, a straightforward calculation shows that the Wigner distribution
$W_{\Psi}$ is (up to an unessential factor), the expectation value of
$\widehat{T}_{\text{GR}}(x_{0},p_{0})$ in the state $|\Psi\rangle$; in fact
(dropping the subscripts $0$)
\begin{equation}
W_{\Psi}(x,p)=\left(  \tfrac{1}{\pi\hbar}\right)  ^{n}\langle\widehat
{T}_{\text{GR}}(x,p)\Psi|\Psi\rangle. \label{wigroyerter}%
\end{equation}
More generally, a similar calculation shows that the cross-Wigner transform is
given by%
\begin{equation}
W_{\Psi,\Phi}(x,p)=\left(  \tfrac{1}{\pi\hbar}\right)  ^{n}\langle\widehat
{T}_{\text{GR}}(x,p)\Phi|\Psi\rangle\label{wigroyer}%
\end{equation}
and can hence be viewed as a transition amplitude. Taking (\ref{5}) into
account we thus have%
\begin{equation}
W_{\Psi,\Phi}(x,p)=(\pi\hbar)^{n}\langle\widehat{T}_{\text{GR}}(x,p)\rangle
_{\Psi,\Phi}\langle\Phi|\Psi\rangle; \label{woofy1}%
\end{equation}
this relation immediately implies, using definition (\ref{rho1}) of the
complex probability distribution $\rho_{\Phi,\Psi}$, the important equality%
\begin{equation}
\rho_{\Phi,\Psi}(x,p)=(\pi\hbar)^{n}\langle\widehat{T}_{\text{GR}}%
(x,p)\rangle_{\Psi,\Phi} \label{woofy2}%
\end{equation}
which can in principle be used to determine $\rho_{\Phi,\Psi}$.

As already mentioned, the cross-ambiguity function $A_{\Psi,\Phi}$ is
essentially the Fourier transform of $W_{\Psi,\Phi}$; in fact
\begin{equation}
A_{\Psi,\Phi}=\mathcal{F}_{\sigma}W_{\Psi,\Phi}\text{ \ , \ }W_{\Psi,\Phi
}=\mathcal{F}_{\sigma}A_{\Psi,\Phi} \label{afw}%
\end{equation}
where $\mathcal{F}_{\sigma}$ is the symplectic Fourier transform: if
$a=a(x,p)$ then $\mathcal{F}_{\sigma}a(x,p)=\widehat{a}(p,-x)$ where
$\widehat{a}$ is the ordinary $2n$-dimensional $\hbar$-Fourier transform of
$a$; explicitly%
\begin{equation}
\mathcal{F}_{\sigma}a(x,p)=\left(  \tfrac{1}{2\pi\hbar}\right)  ^{n}\iint
e^{-\frac{i}{\hbar}(xp^{\prime}-p^{\prime}x)}a(x^{\prime},p^{\prime}%
)d^{n}p^{\prime}d^{n}x^{\prime}. \label{sft}%
\end{equation}
Both equalities in (\ref{afw}) are equivalent because the symplectic Fourier
transform is involutive, and hence its own inverse. While the cross-Wigner
distribution is a measure of \textit{interference}, the cross-ambiguity
function is rather a measure of \textit{correlation}. One shows
\cite{Folland,Birkbis,Springer,hla} that $A_{\Psi,\Phi}$ is explicitly given
by
\begin{equation}
A_{\Psi,\Phi}(x,p)=\left(  \tfrac{1}{2\pi\hbar}\right)  ^{n}\int e^{-\tfrac
{i}{\hbar}py}\Psi(y+\tfrac{1}{2}x)\Phi^{\ast}(y-\tfrac{1}{2}x)d^{n}y.
\label{amboofy2}%
\end{equation}
The cross-ambiguity function is easily expressed using the Heisenberg operator
instead of the Grossmann--Royer operator: we have%
\begin{equation}
A_{\Psi,\Phi}(x,p)=\left(  \tfrac{1}{2\pi\hbar}\right)  ^{n}\langle\widehat
{T}(x,p)\Phi|\Psi\rangle. \label{amboofy1}%
\end{equation}

The following important result shows that the knowledge of the classical
observable $a$ allows us to determine the weak value of the corresponding Weyl
operator using the weak value of the Grossmann--Royer (resp. the Heisenberg) operator:

\begin{proposition}
Let $\widehat{A}$ be the Weyl quantization of the classical observable $a$. We
have%
\begin{equation}
\langle\widehat{A}\rangle_{\Phi,\Psi}=\left(  \tfrac{1}{\pi\hbar}\right)
^{n}\iint a(x,p)\langle\widehat{T}_{\text{GR}}(x,p)\rangle_{\Phi,\Psi}%
d^{n}pd^{n}x \label{weaktgr}%
\end{equation}
and
\begin{equation}
\langle\widehat{A}\rangle_{\Phi,\Psi}=\left(  \tfrac{1}{2\pi\hbar}\right)
^{n}\iint\mathcal{F}_{\sigma}a(x,p)\langle\widehat{T}(x,p)\rangle_{\Phi,\Psi
}d^{n}pd^{n}x. \label{weakthw}%
\end{equation}

\end{proposition}

\begin{proof}
In view of Moyal's formula (\ref{12}) we have%
\begin{equation}
\langle\Phi|\widehat{A}|\Psi\rangle=\iint a(x,p)W_{\Psi,\Phi}(x,p)d^{n}pd^{n}x
\label{Moyal12}%
\end{equation}
that is, taking (\ref{wigroyer}) into account%
\begin{equation}
\langle\Phi|\widehat{A}|\Psi\rangle=\left(  \tfrac{1}{\pi\hbar}\right)
^{n}\iint a(x,p)\langle\widehat{T}_{\text{GR}}(x,p)\Phi|\Psi\rangle
d^{n}pd^{n}x \label{fiapsi}%
\end{equation}
hence (\ref{weaktgr}); formula (\ref{weakthw}) is obtained in a similar way,
first applying the Plancherel formula to the right-hand side of (\ref{Moyal12}%
), then applying the first identity (\ref{afw}), and finally using
(\ref{amboofy1}).
\end{proof}

Notice that the formulas above immediately yield the well-known
\cite{Folland,Birkbis,Springer,Littlejohn} representations of the operator
$\widehat{A}$ in terms of the Grossmann--Royer and Heisenberg operators:%
\begin{equation}
\widehat{A}=\left(  \tfrac{1}{\pi\hbar}\right)  ^{n}\iint a(x,p)\widehat
{T}_{\text{GR}}(x,p)d^{n}pd^{n}x \label{ahat1}%
\end{equation}
and%
\begin{equation}
\widehat{A}=\left(  \tfrac{1}{2\pi\hbar}\right)  ^{n}\iint\mathcal{F}_{\sigma
}a(x,p)\widehat{T}(x,p)d^{n}pd^{n}x. \label{ahat2}%
\end{equation}

\section{The Reconstruction Problem}

\subsection{Lundeen's experiment}

In 2012 Lundeen and his co-workers \cite{Lundeen} determined the wavefunction
by weakly measuring the position, and thereafter performing a strong
measurement of the momentum. They considered the following experiment on a
particle: a weak measurement of $x$ is performed which amounts to applying the
projection operator $\widehat{\Pi}_{x}=|x\rangle\langle x|$ to the
pre-selected state $|\Psi\rangle$; thereafter they perform a strong
measurement of momentum, which yields the value $p_{0}$. The result of the
weak measurement is thus%
\begin{equation}
\langle\widehat{\Pi}_{x}\rangle^{\Psi,\Phi}=\frac{\langle p_{0}|x\rangle
\langle x|\Psi\rangle}{\langle p_{0}|\Psi\rangle}=\left(  \frac{1}{2\pi\hbar
}\right)  ^{n/2}\frac{e^{-\frac{i}{\hbar}p_{0}x}\Psi(x)}{\widehat{\Psi}%
(p_{0})} \label{13.1}%
\end{equation}
($\widehat{\Psi}$ the Fourier transform of $\Psi$). Since the value of $p_{0}$
is known we get
\begin{equation}
\Psi(x)=\frac{1}{k}e^{\frac{i}{\hbar}p_{0}x}\langle\widehat{\Pi}_{x}%
\rangle^{\Psi,\Phi} \label{13.2}%
\end{equation}
where $k=(2\pi\hbar)^{n/2}\widehat{\Psi}(p_{0})$; formula (\ref{13.2}) thus
allows to determine $\Psi(x)$ by scanning through the values of $x$. Thus, by
reducing the disturbance induced by measuring the position and thereafter
performing a sharp measurement of momentum we can reconstruct the wavefunction
pointwise. In \cite{LuBa} Lundeen and Bamber generalize this construction to
mixed states and arbitrary pairs of observables. Using the complex
distribution $\rho_{\Psi,\Phi}(x,p)$ defined above it is easy to recover the
formula (\ref{13.2}) of Lundeen \textit{et al}. In fact, choose $a(x,p)=\Pi
_{x_{0}}(x,p)=\delta(x-x_{0})$; its Weyl quantization
\[
\widehat{\Pi_{x_{0}}}\Psi(x)=\Psi(x_{0})\delta(x-x_{0})
\]
is the projection operator: $\widehat{\Pi}_{x_{0}}|\Psi\rangle=\Psi
(x_{0})|x_{0}\rangle$. Using the elementary properties of the Dirac delta
function together with the marginal property (\ref{rhomarg1}), formula
(\ref{fundamental}) becomes%
\begin{align*}
\langle\Pi_{x_{0}}\rangle_{\Phi,\Psi}  &  =\int\delta(x-x_{0})\rho_{\Phi,\Psi
}(x,p)d^{n}pd^{n}x\\
&  =\int\rho_{\Phi,\Psi}(x_{0},p)d^{n}p\\
&  =\frac{\Phi^{\ast}(x_{0})\Psi(x_{0})}{\langle\Phi|\Psi\rangle}%
\end{align*}
which is (\ref{13.1}); formula (\ref{13.2}) follows.

\subsection{Reconstruction: the WWM approach}

It is well-known \cite{Folland,Birkbis} that the knowledge of the Wigner
distribution $W_{\Psi}$ uniquely determines the state $|\Psi\rangle$; this is
easily seen by noting that $W_{\Psi}$ is essentially a Fourier transform and
applying the Fourier inversion formula, which yields%
\begin{equation}
\Psi(x)\Psi^{\ast}(x^{\prime})=\int e^{\frac{i}{\hbar}p(x-x^{\prime})}W_{\Psi
}(\tfrac{1}{2}(x+x^{\prime}),p)d^{n}p; \label{13.4}%
\end{equation}
one then chooses $x^{\prime}$ such that $\Psi(x^{\prime})\neq0$, which yields
the value of $\Psi(x)$ for arbitrary $x$. The same procedure applies to the
cross-Wigner transform (\ref{9}); one finds that%
\begin{equation}
\Psi(x)\Phi^{\ast}(x^{\prime})=\int e^{\frac{i}{\hbar}p(x-x^{\prime})}%
W_{\Psi,\Phi}(\tfrac{1}{2}(x+x^{\prime}),p)d^{n}p. \label{13.3}%
\end{equation}
Notice that if we choose $x^{\prime}=x$ we recover the generalized marginal
condition (\ref{16}) satisfied by the cross-Wigner distribution.

Thus, the knowledge of $W_{\Psi,\Phi}$ and $\Phi$ is in principle sufficient
to determine the wavefunction $\Psi$. Here is a stronger statement which shows
that the state $|\Psi\rangle$ can be reconstructed from $W_{\Psi,\Phi}$ using
an \textit{arbitrary} auxiliary state $|\Lambda\rangle$ non-orthogonal to
$|\Phi\rangle$:

\begin{proposition}
Let $\Lambda$ be an arbitrary vector in $L^{2}(\mathbb{R}^{n})$ such that
$\langle\Phi|\Lambda\rangle$ $\neq0$. We have
\begin{equation}
\Psi(x)=\frac{2^{n}}{\langle\Phi|\Lambda\rangle}\iint W_{\Psi,\Phi
}(y,p)\widehat{T}_{\text{GR}}(y,p)\Lambda(x)d^{n}pd^{n}y; \label{13.24}%
\end{equation}
equivalently,
\begin{equation}
\Psi(x)=2^{n}\frac{\langle\Psi|\Phi\rangle}{\langle\Phi|\Lambda\rangle}%
\iint\rho_{\Psi,\Phi}(y,p)\widehat{T}_{\text{GR}}(y,p)\Lambda(x)d^{n}pd^{n}y.
\label{100170}%
\end{equation}

\end{proposition}

\begin{proof}
By a standard continuity and density argument it is sufficient to assume that
$\Psi,\Phi,\Lambda$ are in $\mathcal{S}(\mathbb{R}^{n})$. Using the equality
(\ref{13.3}) we have
\[
\Psi(x)\langle\Phi|\Lambda\rangle=\iint e^{\frac{i}{\hbar}p(x-x^{\prime}%
)}W_{\Psi,\Phi}(\tfrac{1}{2}(x+x^{\prime}),p)\Lambda(x^{\prime})d^{n}%
pd^{n}x^{\prime}.
\]
Setting $y=\tfrac{1}{2}(x+x^{\prime})$ we get%
\[
\Psi(x)\langle\Phi|\Lambda\rangle=2^{n}\iint e^{\frac{2i}{\hbar}p(x-y)}%
W_{\Psi,\Phi}(y,p)\Lambda(2y-x)d^{n}pd^{n}y
\]
which yields (\ref{13.24}) in view of the explicit formula (\ref{GRbis}) for
the Grossmann--Royer parity operator.
\end{proof}

\section{Discussion and Perspectives}

We have been able to give a complete characterization of the notion of weak
value in terms of the Wigner distribution, which is intimately related to the
Weyl quantization scheme through Moyal's formula (\ref{7}). There are however
other possible physically meaningful quantization schemes; the most
interesting is certainly that of Born--Jordan \cite{BJ,BJH}, which plays an
increasingly important role in quantum mechanics and in time-frequency
analysis \cite{bog1,bog2,cogoni1,Found,Springer,abreu,gogophysa,gogoletta},
and each of these leads to a different phase space formalism, where the Wigner
distribution has to be replaced by more general element of the
\textquotedblleft Cohen class\textquotedblright\ \cite{23,27}. Unexpected
difficulties however arise, especially when one deals with the reconstruction
problem; these difficulties have a purely mathematical origin, and are related
to the division of distributions (for a mathematical analysis of the nature of
these difficulties, see \cite{cogoni1}). The reconstruction problem for
general phase space distributions will be addressed in a forthcoming publication.

\begin{acknowledgement}
The second author has been funded by the grant P-27773 of the Austrian
Research Agency FWF.
\end{acknowledgement}

\end{document}